\documentclass[a4paper,12pt]{article}

\usepackage{amsmath, amssymb}
\usepackage{bm}
\usepackage{mathtools}
\usepackage{amsthm}
\mathtoolsset{showonlyrefs=true}

\usepackage{geometry}

\newcommand\RR{\mathbb{R}}
\newcommand\ZZ{\mathbb{Z}}
\newcommand\QQ{\mathbb{Q}}
\newcommand\CC{\mathbb{C}}

\newcommand\poly{\mathrm{poly}}

\newtheorem{theorem}{Theorem}

\newtheorem{lemma}[theorem]{Lemma}

\newtheorem{corollary}[theorem]{Corollary}
\newtheorem{example}[theorem]{Example}

\title{Counting Integral Points in Polytopes via Numerical Analysis of Contour Integration}

\author{Hiroshi Hirai, Ryunosuke Oshiro, and Ken'ichiro Tanaka \\
Department of Mathematical Informatics, \\
Graduate School of Information Science and Technology,   \\
The University of Tokyo, Tokyo, 113-8656, Japan.\\
\texttt{\normalsize hirai@mist.i.u-tokyo.ac.jp}\\
\texttt{\normalsize ryunosuke$\_$oshiro@mist.i.u-tokyo.ac.jp} \\
\texttt{\normalsize kenichiro@mist.i.u-tokyo.ac.jp}
}

\begin{document}

\maketitle

\begin{abstract}
    In this paper, we address the problem of counting integer points 
    in a rational polytope described by $P(y) = \{ x \in \RR^m \colon Ax = y, x \geq 0\}$, where $A$ is an $n \times m$ integer matrix and $y$ is an $n$-dimensional integer vector.
    We study the Z-transformation approach initiated by 
    Brion-Vergne, Beck, and Lasserre-Zeron from the numerical analysis point of view, and obtain a new algorithm on this problem: 
    If $A$ is nonnegative, then the number of integer points in $P(y)$ 
    can be computed in $O(\poly (n,m, \|y\|_\infty)  (\|y\|_\infty + 1)^n)$ 
    time and $O(\poly (n,m, \|y\|_\infty))$ space.
    This improves, in terms of space complexity, a naive DP 
    algorithm with $O((\|y\|_\infty + 1)^n)$-size DP table.
    Our result is based on the standard error analysis 
    to the numerical contour integration for the inverse Z-transform, and 
    establish a new type of an inclusion-exclusion formula 
    for integer points in $P(y)$.
    
    We apply our result to hypergraph $b$-matching, 
    and obtain a $O(\poly(\allowbreak n,m,\allowbreak\|b\|_\infty)\allowbreak (\|b\|_\infty +1)^{(1-1/k)n})$
    time algorithm for counting $b$-matchings 
    in a $k$-partite hypergraph with $n$ vertices and $m$ hyperedges.
    This result is viewed as 
    a $b$-matching generalization of the classical result by Ryser for $k=2$ 
    and its multipartite extension by Bj{\"o}rklund-Husfeldt. 
\end{abstract}

Keywords: Integer points in polytopes, counting algorithm, Z-transformation, numerical integration, trapezoidal rule

\section{Introduction}

    Counting integer points in 
    polytopes is a fundamental problem. 
    There are numerous applications in various areas of mathematical science, 
    and fascinating mathematics behind; 
    see e.g., \cite{barvinok_book,beck_book}.
    This problem is computationally intractable, i.e., it is  $\#$P-hard~\cite{valiant1979complexity}.  
    Approximate counting as well as 
    exact counting under fixed parameter settings 
    has been rich sources for developments 
    in the theory of algorithms and computational complexity.
    A seminal work by Barvinok~\cite{barvinok1994} showed that there is a polynomial time algorithm to count integer points in rational polytope $P$ when the dimension $d$ of $P$ 
    is fixed.  
    His algorithm computes 
    a certain ``compact'' expression ({\em Brion-Lawrence formula}) of 
    the generation function 
    $g_P(z) \coloneqq \sum_{a \in P \cap \ZZ^d} z_1^{a_1} z_2^{a_2} \cdots z_d^{a_d}$ 
    for multivariate indeterminate $z = (z_1,z_2,\ldots,z_d)$.
    This (extremely difficult) algorithm is now implemented in computer package {\em LattE}~\cite{DELOERA20041273}, and provides a useful tool 
    to the study of geometric combinatorics. 

    A ``dual'' generating-function approach
    was initiated by Brion-Vergne~\cite{brion1997residue}, Beck~\cite{beck2000counting}, and
    Lasserre-Zeron~\cite{LasserreZeron2003,LasserreZeron2005}; see \cite{LasserreBook}. 
    Suppose now that the input polytope $P = P (y)$ is given by
    \begin{equation}
    P(y) = \{ x \in \RR^m \colon Ax = y, x \geq 0\},
    \end{equation}
     for $n \times m$ integer matrix $A$ and $n$-dimensional vector $y \in \ZZ^n$.
     Let $f_A(y) \coloneqq |P(y) \cap \ZZ^m|$, 
     and consider the {\em Z-transform} $\hat f_A(z) \coloneqq \sum_{y \in \ZZ^n} f_A(y) z_1^{y_1} z_2^{y_2}\cdots z_n^{y_n}$.
     Brion-Vergne~\cite{brion1997residue} showed that $\hat{f}_A$ admits a very simple closed formula $\hat f(z)\!\! =\!\! \prod_{k=1}^{m} 1/(1- z_1^{A_{1k}}z_2^{A_{2k}} \cdots z_n^{A_{nk}})$, and 
     that the wanted $f_A(y)$ is recovered by the {\em inverse Z-transformation}, 
     which is a multi-dimensional contour integration of $\hat{f}_A$.
     This reduces the counting problem to the residue computation of $\hat{f}_A$. 
     By this approach, 
     Lasserre-Zeron~\cite{LasserreZeron2003,LasserreZeron2005} developed an $O((n+1)^{m-n}\Lambda)$-time algorithm to count integral points in $P(y)$, 
     where $\Lambda$ is a function of matrix $A$.
     
     In this paper, we study the contour integration of the inverse Z-transformation 
     from the numerical analysis point of view, and obtain 
     a new algorithm
     to count integer points for 
     an important class of polytopes.
     Our main result is as follows.
\begin{theorem}\label{thm:main}
	Suppose that $A$ is nonnegative.
	For $y \in \ZZ^n$, 
	the number of integer points in the polytope $P(y)$
	can be computed in $O (\poly(n,m,\|y\|_\infty) (\|y\|_\infty +1)^n)$ 
	time and $O (\poly(n,m,\|y\|_\infty))$ space.
\end{theorem} 
       Notice that there is a simple DP algorithm with the same time complexity:
       For $k=1,2,\ldots,m$, consider the matrix $A^k$ consisting of the first $k$ columns of $A$, and the number 
       $N^k(z)$ of integer points of polyhedron $\{x \in \RR^k \mid A^k x= z, x \geq 0\}$.
       By the nonnegativity of $A$, 
       integers $N^{k+1}(z)$ for $0 \leq z \leq y$ are 
       obtained from $N^{k}(z)$ for $0 \leq z \leq y$ in $O(\poly(n,\|y\|_\infty)(\|y\|_\infty +1)^n))$ time. The resulting DP algorithm, however, requires an 
       $O(\|y\|_\infty +1)^n)$-space for the DP table. 
       Thus our result is regarded as an improvement in terms of space complexity.

       Our technique for proving Theorem~\ref{thm:main} is as follows.
       Instead of the residue computation of $\hat f_A$, 
       we apply the numerical integration to 
       the inverse Z-transform of $\hat f_A$, where we use the {\em trapezoidal rule}, a basic and popular method of numerical integration. 
       By the standard error analysis of the trapezoidal rule~\cite{trefethen2014exponentially}, 
       we obtain 
       an error estimate with respect to 
       the number $N$ of sampled points of the numerical integration 
       and contour radius $r$.  
       Interestingly this estimate gives rise to a new inclusion-exclusion type formula
       for $f_A(y)$, and brings the algorithm in Theorem~\ref{thm:main}, which is quite simple and is easier to be implemented. 
       A notable feature of our inclusion-exclusion is 
       to use the cancellation structure of 
       trigonometric function $\exp (2\pi i k/N)$ in the complex plane $\CC$.
       This extends the usual inclusion-exclusion based on 
       the cancellation of $1$ and $-1$ in $\RR$. 
       Our algorithm computes the number of integer points in  
       the expression $\sum_{k=0}^{N-1} a_k \exp (2\pi i k/N)$ for $a_k \in \QQ$, 
       and recovers the ``true'' value 
       by algebraic computation on the group ring of cyclic group $\ZZ/N\ZZ$, 
       which avoids numerical computation of $\exp (2\pi i k/N)$.

    Our result is applicable 
    to packing-type polytopes, which ubiquitously arise 
    from graph theory and combinatorial mathematics.
    Consider the particular case 
    where $A$ is 0-1 valued and $y$ is the all-one vector ${\bf 1}$.
    Then $f_A({\bf 1}) = |P({\bf 1}) \cap \ZZ^m|$ is the number of
	perfect matchings in the hypergraph corresponding to $A$.
	Exact exponential time counting algorithms for matchings 
	have been intensively studied in recent years~\cite{bjorklund2012counting,bjorklund2008exact,cygan2015faster,izumi2012new}.
	The formula of hypergraph matching derived from our result 
    is viewed as a variant of that given by Bj{\"o}rklund-Husfeldt~\cite{bjorklund2008exact}. 
	The time complexity $O^*(2^n)$ matches that of their algorithm.
	By exploiting special properties, 
    Bj{\"o}rklund-Husfeldt~\cite{bjorklund2008exact} 
	improve the time complexity to $O^* (2^{(1-1/k)n})$ for $k$-partite hypergraphs.
	This result is viewed as a multipartite extension of 
	the classical result of Ryser~\cite{ryser1963} for bipartite matching.
	We also exploit a special structure of our formula 
	in $k$-partite hypergraphs, and prove
    the following $b$-matching generalization of the results of Ryser and Bj{\"o}rklund-Husfeldt.
     \begin{theorem}\label{thm:main2}
     	Let ${\cal H} = (V,{\cal E})$ be a $k$-partite hypergraph, 
     	and let $b: V \to \ZZ_+$. 
     	The number of perfect $b$-matchings in ${\cal H}$
     	can be computed in $O(\poly (|V|,|{\cal E}|, \allowbreak\|b\|_\infty) (\|b\|_\infty+1)^{(1-1/k)|V|})$ time and $O(\poly (|V|,|{\cal E}|, \|b\|_\infty)$ space.
     \end{theorem}
   Counting perfect $b$-matchings of a $k$-partite graph
   (i.e., $k$-dimensional $b$-matchings) has many
   applications in a wide range of mathematical sciences that include
   combinatorics,
   representation theory, and statistics; see e.g.,~\cite{DiaconisGangolli1993}.
   For example, counting multiway contingency tables with prescribed margins,
   an important problem for statistical analysis on contingency tables,
   is nothing but $k$-dimensional $b$-matching counting.

   The rest of this paper is organized as follows.
   In Section~\ref{sec:pre}, we set up basic notation, and 
   introduce Z-transformation, its inverse, and approximate inverse Z-transformation  
   obtained by numerical integration.
   In Section~\ref{sec:algo}, we present our algorithm to prove the main theorem.
   In Section~\ref{sec:match}, 
   we discuss hypergraph matching and prove Theorem~\ref{thm:main2} 
   in a further generalized form.

\section{Preliminaries}\label{sec:pre}
\label{preliminaries}
\subsection{Notation}
Let $\mathbb{R}_{>0}$ denote the set of positive real numbers.
Let $\mathbb{Z}_{\geq 0}$ denote the set of nonnegative integers.
For a matrix $A\in \mathbb{Z}^{n\times m}$, let $A_k$ denote the $k$-th column vector of the matrix.
For an integer vector $y=\left( y_1,\dots,y_n \right)^{\top}\in \mathbb{Z}^n$ and a complex vector 
$z = \left( z_1 ,\dots,z_n\right)^{\top}\in \mathbb{C}^n$,
define $z^y \in \mathbb{C}$ by
\[
    z^y \coloneqq z_1^{y_1}\cdots z_n^{y_n}.
\]
For a positive integer $N > 0$, define $\omega_N: \{0,1,2,\ldots,N-1\} \to \CC$ by
\begin{align}
\omega_N(h) \coloneqq \exp \frac{2\pi i}{N} h \quad (h \in \{0,1,2,\ldots,N-1\}).
\end{align}
The following relation is well-known:
\begin{align}
\sum_{j=0}^{N-1} \omega_{N}(kj) = \left\{
\begin{array}{ll}
N & {\rm if}\ k= 0 \bmod N, \\
0 & {\rm otherwise},
\end{array}\right. \quad (k \in \ZZ).\label{eqn:omega}
\end{align}
For a function $f: D \to \CC$ and $z \in D^n$, 
let $f(z)$ denote $(f(z_1),f(z_2),\ldots,f(z_n))^{\top}$, 
such as 
\begin{align}
\exp (z) &= (\exp (z_1),\exp (z_2),\ldots,\exp (z_n))^{\top}, \\
\mathrm{ln}\, z &= \left( \mathrm{ln}\, z_1, \dots , \mathrm{ln}\, z_n\right)^{\top},\\
\omega_N(z) & = (\omega_N(z_1),\ldots, \omega_{N}(z_n))^{\top}. 
\end{align}

The symbol $\bot$ is meant as ``undefined.''
We use $\bot$ 
when the function value is defined via integration or 
infinite summation, possibly not converging.

\subsection{Z-transformation}
For a function $f\colon \mathbb{Z}^n \to \mathbb{R}$, 
define the {\em Z-transform} $\hat{f}\colon \mathbb{C}^n\to \mathbb{C}\cup \left\{\bot\right\}$ of $f$
by
\begin{align}
\hat{f}(z) \coloneqq \sum_{y\in\mathbb{Z}^n} f( y ) z^{-y} \quad (z \in \mathbb{C}^n).\label{eqn:hatf}
\end{align}
The inverse of the Z-transformation is given as follows.
For a function $g\colon \mathbb{C}^n \to \mathbb{C} \cup \{\bot \}$ and $r > 0$,
define $I_r[g] \colon \mathbb{Z}^n \to \mathbb{R}\cup \left\{ \bot \right\}$ by
\begin{align}
I_r[g] ( y ) &\coloneqq \frac{1}{\left(2\pi i\right)^n} \oint_{\lvert z_1 \rvert = r} \dots \oint_{\lvert z_n \rvert = r} g\left(z\right) 
z^{y- \bm{1}} \mathrm{d} z_1 \cdots \mathrm{d} z_n \\ 
&= \int_{[0,1)} \dots \int_{[0,1)} g\left(r\exp\left(2\pi i t\right)\right) 
r^{{\bf 1}^{\top}y}\exp\left( 2\pi i t^{\top} y\right) \mathrm{d} t_1 \cdots \mathrm{d}t_n
\quad (y \in \ZZ^n),\label{eqn:integration2}
\end{align}
where we change variables by $z_k = r \exp (2\pi i t_k)$ in \eqref{eqn:integration2}.
Under an appropriate condition on $f$ and $r$, 
map $g \mapsto I_r[g]$ is actually the inverse of the Z-transformation:
\begin{equation}\label{eqn:I_r[hatf]=f}
I_r[\hat f] = f. 
\end{equation}
We do not go into details under which conditions \eqref{eqn:I_r[hatf]=f} holds. 
Instead, we consider an approximate inverse Z-transform 
by the numerical integration applied to (\ref{eqn:integration2}).
Here we use the {\em trapezoidal rule}, 
which is a basic and popular method of numerical integration; see e.g.,~\cite{QuadratureBook}.
For a positive integer $N > 0$ (the number of points 
in the numerical integration), 
define $I_{N,r} [ g ]: \ZZ^n \to \CC \cup \{\bot\}$ by
\begin{equation}
    I_{N,r}[ g ]  ( y ) \coloneqq \frac{1}{N^n} \sum_{j\in \{ 0,1,\ldots, N-1\}^n }
    g(r\omega_N(j)) r^{{\bf 1}^{\top} y} 
    \omega_N(j^{\top} y) \quad (y \in \ZZ^n). \label{eqn:I_N,r}
\end{equation}
Recall notation $\omega_N(j)\! \coloneqq\! \exp (2\pi i j/N) 
\! = \! (\exp (2\pi j_1/N), \exp (2\pi j_2/N),\allowbreak\ldots,\allowbreak \exp (\allowbreak 2\pi j_n/N)\allowbreak )$.
Our counting algorithm is based on $I_{N,r}$.

\section{Counting integral points in a polytope}\label{sec:algo}
Let $A$ be an $n\times m$ integral matrix. We assume that 
there is no nonzero nonnegative vector $x \in \ZZ^{m}_{\geq 0} \setminus \{0\}$
with $Ax = 0$.
This assumption ensures that the polytope $\{ x \in \RR^m \colon Ax = y, x \geq 0\}$
is bounded for every $y \in \ZZ^n$.
In the case where $A$ is a nonnegative matrix, 
this assumption is equivalent to the property that each column of $A$ has at least one nonzero entry. 

As mentioned in the introduction, 
define function $f_A: \ZZ^n \to \ZZ$ by
\[
    f_A\left(y\right) \coloneqq |\left\{ x \in \mathbb{Z}^m \colon Ax = y, x\geq 0 \right\}|
    \quad (y \in \ZZ^n).
\]
Our starting point is the following formula for $\hat f_{A}$.
\begin{theorem}[\cite{beck2000counting,brion1997residue,LasserreZeron2003}]
    \label{theorem:lz}
    \begin{enumerate}
        \item For $z \in \CC^n$ with  $A^{\top} \mathrm{ln}\, |z| > 0$,  
            the Z-transform $\hat{f}_A(z)$ is given by
            \begin{equation}
                \hat{f}_A\left(z\right) = \sum_{h \in \ZZ_{\geq 0}^m} z^{-Ah} = \prod_{k=1}^m \frac{1}{1-z^{-A_{k}}},
            \end{equation}
            where the series absolutely converges.
        \item
            For $s \in \RR^n_{> 0}$ with $A^\top  \mathrm{ln}\, s > 0$, 
            it holds
            \begin{equation}
                f_A\left( y \right) = \frac{1}{\left( 2 \pi i \right)^n}\oint_{\lvert z_1\rvert = s_1} \cdots \oint_{\lvert z_n\rvert = s_n}
                \hat{f}_A \left( z \right) z^{y-\bm{1}}\,\mathrm{d}z_1\cdots\mathrm{d}z_n.
            \end{equation}
   \end{enumerate}
\end{theorem}
We establish 
an approximate version of the above theorem as follows. 
\begin{theorem}
    \label{thm:estimate}
    Suppose that $A$ is nonnegative.
    For $y \in \ZZ_{\geq 0}^n$, $r > 1$, and $N \coloneqq \| y \|_{\infty}+1$, it holds
    \begin{align}
        I_{N,r} [ \hat{f}_A]\left(y\right) 
        &= \frac{1}{N^{n}} \sum_{h \in \ZZ_{\geq 0}^m}
        \sum_{j\in \left\{ 0,1,\ldots,N-1\right\}^n}
        r^{{\bf 1}^{\top}(y - A h)}\omega_N ( j^{\top}(y - A h))  \label{eqn:estimate1}
        \\
        &= f_A\left( y\right) +  
        \sum_{k=1}^\infty r^{-Nk} 
        |\{ x \in \ZZ^m_{\geq 0} \colon y - Ax \geq 0, {\bf 1}^{\top} (y - Ax) = Nk \}|. \label{eqn:estimate2}
    \end{align}
\end{theorem}

\begin{proof}
    From the assumption that $A$ is nonnegative and each column of $A$ has nonzero entry, 
    for $z \in \CC^n$ with $|z_1| = |z_2| = \cdots = |z_n| = r > 1$, it holds $A^{\top} \ln |z| = A^{\top} \ln r > 0$.
    By the previous theorem, 
    the Z-transform $\hat{f}_A(z)$ is given by 
    \begin{align}
    \hat{f}_A(z) = \sum_{h \in \ZZ^m_{\geq 0}} z^{- Ah}.
    \end{align}
    Substituting this expression to $I_{N,r}[\hat f_A](y)$ (in \eqref{eqn:I_N,r}), we have
    \begin{align}
    I_{N,r}[ \hat{f}_A ]\left( y \right)
    &= \frac{1}{N^{n}}
    \sum_{j \in \left\{ 0,1,\ldots,N-1\right.\}^n} \sum_{h \in \ZZ^m_{\geq 0}}
    r^{{\bf 1}^\top(y -Ah)} \omega_N(j^{\top}(y -Ah)) \\
    &= \frac{1}{N^{n}}
    \sum_{h \in \ZZ^m_{\geq 0}} r^{{\bf 1}^\top(y -Ah)} \sum_{j \in \left\{ 0,1,\ldots,N-1\right\}^n} 
     \omega_N(j^{\top}(y -Ah)) \\
     &= \frac{1}{N^{n}}
     \sum_{h \in \ZZ^n_{\geq 0}} r^{{\bf 1}^\top(y -Ah)} \prod_{l = 1}^m 
     \sum_{j_{l}=0}^{N-1} 
     \omega_N(j_{l}(y_l -(Ah)_l)), 
    \end{align}
    where the summations are interchangeable, 
    thanks to the absolute convergence (by $r > 1$).
    By the relation \eqref{eqn:omega},
    $\prod_{l=1}^m \sum_{j_{l}=0}^{N-1}\omega_N(j_{l}(y_l -(Ah)_l))$ is $N^n$ if 
    $Ah - y \in N \ZZ_{\geq 0}^n$, and zero otherwise. 
    Notice that $(A h - y)_l \leq - N$ cannot occur since $\|y\|_{\infty} < N$. 
    Thus we have 
    \[
        I_{N,r} [ \hat{f}_A]\left(y\right)  =  \sum_{u \in \ZZ^n_{\geq 0}} r^{-N{\bf 1}^{\top} u}  
    |\{ x \in \ZZ^m_{\geq 0} \colon y - Ax = Nu\}|.
    \]
    Gathering $u \in \ZZ_{\geq 0}^n$ with ${\bf 1}^{\top} u =k$, we obtain \eqref{eqn:estimate2}.    
\end{proof}

This proof is inspired by 
the standard argument 
to derive the exponential convergence
of the trapezoidal rule applied to periodic functions, where
the the cancellation technique using (\ref{eqn:omega}) in the proof 
is known as {\em aliasing} in numerical analysis;
see \cite[Theorem 2.1]{trefethen2014exponentially}.

Notice that $g(k) \coloneqq |\{ x \in \ZZ^m_{\geq 0} \colon y - Ax \geq 0, {\bf 1}^{\top} (y - Ax) = k \}|$ is a variant of the Ehrhart (quasi)polynomial, 
and is bounded by a polynomial in $k$ with degree $m - n$. 
Hence the series in~\eqref{eqn:estimate2} actually absolutely converges for $r > 1$.

\begin{corollary}\label{corollary:integerpointsinpolytopes}
    Suppose that $A$ is nonnegative.
    Let $y \in \ZZ^n_{\geq 0}$ and $N \coloneqq \lVert y \rVert_{\infty} + 1$. 
    Then $f_A\left( y \right)$ is equal to the coefficient of $r^{- \bm{1}^{\top} y }$ in
    \begin{align}
    	& \frac{1}{N^n}
    	\sum_{j\in \left\{ 0,1,\ldots,N-1 \right\}^n} \sum_{h \in \{0,1,2,\ldots,N-1\}^m}
    	r^{- \bm{1}^{\top} Ah} \omega_N(j^{\top} (y- Ah) )  \label{poly_formula0} \\ 
        & = \frac{1}{N^n}
        \sum_{j\in \left\{ 0,1,\ldots,N-1 \right\}^n}
        \omega_N(j^{\top} y)
        \prod_{l=1}^m\sum_{h_l = 0}^{N-1}r^{- \bm{1}^{\top} A_{l}h_l}
        \omega_N(- j^{\top} A_{l}h_l) \label{poly_formula}.
    \end{align}
\end{corollary}
Here we regard $r$ as an indeterminate.
\begin{proof}
	In the formula \eqref{eqn:estimate1}, 
	if $h \in \ZZ^m_{\geq 0}$ has $h_i \geq N$, then $(Ah)_l > y_l$ for some $l$, 
	and $h$ does not contribute to the constant term $f_A(y)$.
	The claim follows from this fact and 
	the observation $\sum_{h \in \{0,1,\ldots,N-1\}^n} r^{- {\bf 1}Ah} \omega_N(-j^{\top} Ah) 
	= \prod_{l=1}^m\sum_{h_l = 0}^{N}r^{- \bm{1}^{\top} A_{l}h_l}
	\omega_N(-j^{\top} A_{l}h_l)$.
\end{proof}
Our goal is to compute the coefficient of $r^{- {\bf 1}^{\top}y}$ in \eqref{poly_formula}.
Instead of numerical computation of the trigonometric function $\omega_N$, 
we develop an algebraic algorithm.
Let $\QQ[\ZZ/N\ZZ]$ denote the group ring of 
the cyclic group $\ZZ/N\ZZ = \{0,1,2,\ldots,N-1\}$ of order $N$.
Namely, $\QQ[\ZZ/N\ZZ]$ consists of polynomials with variable $s$,  rational coefficients, and degree at most $N-1$, 
in  which the multiplication rule is given by $s^{l} \cdot s^{l'} = s^{l+l' \bmod N}$.
Consider the bivariate polynomial ring $\QQ[\ZZ/N\ZZ][t]$ with variables $s,t$. 
Then $p(s,t) \mapsto p(\omega_N(1),t)$ is a ring homomorphism 
from $\QQ[\ZZ/N\ZZ][t]$ to $\CC[t]$.
Letting $t = r^{-1}$ and $s = \omega_N(1)$ in \eqref{poly_formula0}, 
we obtain a polynomial in $\QQ[\ZZ/N\ZZ][t]$: 
\begin{equation}
\frac{1}{N^n}
\sum_{j\in \left\{ 0,1,\ldots,N-1 \right\}^n} \sum_{h \in \{0,1,\ldots,N-1\}^m} 
t^{\bm{1}^{\top} Ah} s^{j^{\top} (y- Ah)}. \label{eqn:ts}
\end{equation}
Let $\overline{f_A}(s) \in \QQ[\ZZ/N\ZZ]$ denote the coefficient 
of $t^{\bm{1}^{\top} y}$ in \eqref{eqn:ts}.
Then it holds
\begin{equation}
\overline{f_A}(\omega_N(1)) = f_A(y).
\end{equation}
Our algorithm first computes $\overline{f_A}(s)$, and then computes $\overline{f_A}(\omega_N(1)) = f_A(y)$.
\begin{lemma}
    $\overline{f_A}(s) = a_0 + a_1 s + \cdots + a_{N-1} s^{N-1}$ can be computed in $O(\poly (m, n,\allowbreak N) N^n)$ time and $O(\poly (n,m,N))$ space.
\end{lemma}
\begin{proof}
	Let $d \coloneqq {\bf 1}^\top y$.
	From \eqref{poly_formula}, we first consider the computation of 
	the coefficient $b_j(s)$ of $t^{d}$ in
	\begin{equation}
	\prod_{l=1}^m\sum_{h_l = 0}^{N-1}t^{\bm{1}^{\top} A_{l}h_l}
	s^{- j^{\top} A_{l}h_l} \label{eqn:prod} 
	\end{equation}
	for fixed $j$.
    It suffices to compute the above polynomial \eqref{eqn:prod} modulo $(t^{d+1})$, which can be written as
    \begin{equation}
    \prod_{l=1}^m s^{k^l_0} + s^{k^l_1} t^1 + \cdots + s^{k^l_{d}} t^{d} \mod (t^{d+ 1}). \label{eqn:prod2} 
    \end{equation}
    for some $k^l_i \in \{0,1,2,\ldots,N-1\}$.	
    The integers $k^l_i$ are obtained in $O(\poly (n,m,N\allowbreak ))$ time by computing $\bm{1}^{\top} A_l h_l (\leq d)$ and $j^{\top} A_l h_l (\bmod N)$ for 
    $l =1,2,\ldots,m, h_l =0,1,2,\ldots,N-1$.
    Expand \eqref{eqn:prod2} to the form 
    $\alpha_0(s) + \alpha_1(s) t^1 + \cdots + \alpha_{d}(s) t^{d}$, 
    and obtain $b_j(s) = \alpha_{d}(s)$. 
    This computation is done by the multiplication
    of $m$ polynomials with degree $d$ modulo $(t^{d+1})$, where their coefficients are polynomials with degree $N-1$ in multiplication rule $s^{l} s^{l'} = s^{l + l' \bmod N}$.
    Now $\overline{f_A}(s)$ is the sum of 
    $s^{j^{\top}y} b_j(s)$ over $j \in \{ 0,1,\ldots,N-1 \}^n$ divided by $N^n$.
    Thus $\overline{f_A}(s)$ is obtained in $O(\poly (n,m,N) N^n)$ arithmetic operations (over $\ZZ$).
    
    Finally we estimate the bit-size required for the computation.
    It suffices to estimate the size of $b_j(s)$, 
    which is the sum of at most ${d+m-1 \choose d}$ terms of form $s^{k}$. 
    Then the coefficients of $b_j(s)$ have bit-length $O(\poly (n,m,N))$.
    Thus the required bit-size is at most $\poly (n,m,N) n \log N$. 
\end{proof}
Next we consider how to compute
$\overline{f_A}(\omega_N(1)) = f_A(y)$ from $\overline{f_A}(s)$.
\begin{lemma} $\overline{f_A}(s)$ is written as
	\[
	\overline{f_A}(s) = f_A(y) + \sum_{i} K_i(1 + s^i + s^{2i} + \cdots + s^{N-i}).
	\]
	where the sum is taken over divisors $i < N$ of $N$ with some coefficient $K_i \in \QQ$. 
\end{lemma}
\begin{proof}
	Regard $\{0,1,\ldots,N-1\}^n$ as $(\ZZ/N\ZZ)^n$.
	Then the map $\varphi_h :(\ZZ/N\ZZ)^n \to \ZZ/N\ZZ$ 
	defined by $j \mapsto j^{\top} (y - Ah) \bmod N$ is a 
	group homomorphism.
	Therefore the image of $\varphi_h$ is the cyclic group $n_h \ZZ/N \ZZ$ 
	for some divisor $n_h$ of $N$. 
	Also the number of inverse images of each $k \in \{0,1,\ldots,N/n_h-1\}$ is given by $J_h \coloneqq |\ker \varphi_h|$.
	Then,  for $h \in \{0,1,\ldots,N-1\}^m$ with  
	${\bf 1}^{\top}Ah  = {\bf 1}^{\top}y$, it holds
		\begin{align}
		& \sum_{j\in \left\{ 0,1,\ldots,N-1 \right\}^n}
		t^{\bm{1}^{\top} Ah} s^{j^{\top} (y- Ah)}  \\
		& = t^{\bm{1}^{\top} y} \sum_{j\in \left\{ 0,1,\ldots,N-1 \right\}^n}
		s^{j^{\top} (y- Ah)} = t^{\bm{1}^{\top} y} J_h (1+ s^{n_h} + s^{2n_h} + \cdots + s^{N-n_h}).
		\end{align}
		Thus we have
		\begin{align}
		\overline{f_A}(s) & = \sum_{h \in \{0,1,\ldots,N-1\}^m} J_h (1+ s^{n_h} + s^{2n_h} + \cdots + s^{N-n_h}) \\
		& = \sum_{i: {\rm divisor\ of\ } N} K_{i} (1+ s^{i} + s^{2i} + \cdots + s^{N-i}) , \label{eqn:divisor}
		\end{align}
		where $K_{i}$ is the sum of $J_h$ over $h \in (\ZZ/N\ZZ)^n$ 
		such that the image of $\varphi_h$ is $i\ZZ/N\ZZ$.
		Notice $K_N = f_A(y)$.
\end{proof}
According to this lemma,  we obtain a simple algorithm to 
compute $f_A(y)$ from $\overline{f_A}(s)$ as follows.
\begin{itemize}
	\item[0:] Let $\overline{f_A}(s) = a_0 + a_1s + \cdots + a_{N-1} s^{N-1}$.
	\item[1:] If $a_i = 0$ for all $i > 0$, then output $a_0 = f_A(y)$; stop 
	\item[2:] Choose the minimum index $i > 0$ with $a_i \neq 0$.
	Let $a_j \leftarrow a_j - a_i$ for each index $j$ that is the multiple of the index $i$, and go to step 1.
\end{itemize}
The correctness of the algorithm is clear from the above lemma:
The chosen index $i$ in step 2 is a divisor of $N$ with $a_i = K_i$.
Hence the algorithm computes 
$\overline{f_A}(s) - \sum_{i} K_i (1+ s^{1} + \cdots + s^{N-i})$.   
After at most $N$ iterations, the algorithm terminates
and outputs the correct answer $f_{A}(y)$.

\begin{example}
    Consider the following matrix $A$ and vector $y$:
    \begin{equation}
        A = \begin{pmatrix}
            1 & 1 & 3 \\
            1 & 1 & 1
        \end{pmatrix}
        ,
        y = \begin{pmatrix}
            5 \\
            3
        \end{pmatrix}.
    \end{equation}
    Then the polytope $\{ x\in \RR^3 \colon Ax = y, x\geq 0\}$ has three integer points:
    \begin{equation}
        \begin{pmatrix} 1 \\ 1 \\ 1\end{pmatrix},
        \begin{pmatrix} 2 \\ 0 \\ 1\end{pmatrix},
        \begin{pmatrix} 0 \\ 2 \\ 1\end{pmatrix}.
    \end{equation}
    Let us count the integer points according to Corollary~\ref{corollary:integerpointsinpolytopes}.
    The coefficient $\overline{f_A}(s)$ of $t^{8}$ in
    \begin{equation}
        \frac{1}{6^2} \sum_{j_1=0}^5 \sum_{j_2=0}^5 s^{5j_1 + 3j_3}
        \left(
        \sum_{h_1 = 0}^5 t^{2h_1}s^{(-j_1 -j_2)h_1}
        \right)^2
        \left(
        \sum_{h_2 = 0}^5 t^{4h_2}s^{(-3j_1-j_2)h_2}
        \right)
        \\
    \end{equation}
    is
    \begin{align}
        & \frac{1}{6^2} \sum_{j_1=0}^5 \sum_{j_2=0}^5 s^{5j_1 + 3j_3}
        \left( 5s^{-4j_1-4j_2} + 3s^{-5j_1-3j_2} + s^{-6j_1-j_2} \right) 
        \\
        =&\frac{1}{36} \sum_{j_1=0}^5 \sum_{j_2=0}^5
        \left(
         5 s^{j_1 - j_2} + 3 + s^{-j_1+2j_2}
        \right)
        \\
        =&\frac{1}{36}
        \left(
        120 + 12s + 12 s^2 + 12 s^3 + 12 s^4 + 12 s^5
        \right)
        = 3 + \frac{1}{3}\left( 1+ s + s^2 + s^3 + s^4 + s^5\right).
    \end{align}
    Therefore we obtain $\overline{f_A}(\omega_6(1)) = 3 = f_A(y)$.
\end{example}
\begin{example}
    Consider the following matrix $A$ and vector $y$:
    \begin{equation}
        A = \begin{pmatrix}
            1 & 2 \\
            2 & 1 
        \end{pmatrix}
        ,
        y = \begin{pmatrix}
            7 \\
            5
        \end{pmatrix}
    \end{equation}
    Then the polytope $\{ x\in \RR^3 \colon Ax = y, x\geq 0\}$ has only one integer point:
    \begin{equation}
        \begin{pmatrix} 3 \\ 2 \end{pmatrix}.
    \end{equation}
    Consider the coefficient $\overline{f_A}(s)$ of $t^{12}$ in
    \begin{equation}
        \frac{1}{8^2}\sum_{j_1=0}^7\sum_{j_2=0}^7
        s^{7j_1 + 5j_2}
        \left( \sum_{h_1=0}^7 t^{3h_1}s^{(-2j_1-j_2)h_1}\right)
        \left( \sum_{h_2=0}^7 t^{3h_2}s^{(-j_1-2j_2)h_2}\right).
    \end{equation}
    By calculation, 
    $\overline{f_A}(s)$ is 
    \begin{align}
        &\frac{1}{8^2}
        \sum_{j_1=0}^7\sum_{j_2=0}^7
        s^{7j_1+5j_2}
        \left( s^{-8j_1-4j_2} + s^{-7j_1-5j_2} + s^{-6j_1-6j_2} + s^{-5j_1-7j_2} + s^{-4j_1-8j_2}\right)
        \\
        =& \frac{1}{64}
        \sum_{j_1=0}^7\sum_{j_2=0}^7
        \left( s^{-j_1+j_2} + 1 + s^{j_1-j_2} + s^{2j_1-2j_2} + s^{3j_1-3j_2}\right)
        \\
        =& \frac{1}{64}\left( 104 + 24s + 40s^2 + 24s^3 + 40 s^4 + 24s^5 + 40s^6 + 24s^7\right)
        \\
        =& 1 + 
        \frac{3}{8}\left( 1 + s + s^2 + s^3 + s^4 + s^5 + s^6 + s^7\right)
        + \frac{1}{4}\left(1 + s^2 +s^4 +s^6\right).
    \end{align}
    Thus we obtain $\overline{f_A}(\omega_8(1)) = 1 = f_A(y)$.
\end{example}

\section{Hypergraph Matching}\label{sec:match}

We next show Theorem~\ref{thm:main2} in a generalized form.
Let $A$ be an $n \times m$ nonnegative integer matrix.
For each column index $l \in \{1,2,\ldots, m\}$, 
consider subset $F_l$ consisting of row indices $k \in \{1,2,\ldots,n\}$ 
with $A_{kl} > 0$.
Let ${\cal H}(A)$ denote the hypergraph on vertex set $\{1,2,\ldots,n\}$
and hyperedge set $\{  F_l : l =1,2,\ldots,m\}$.
By a {\em stable} set of ${\cal H}(A)$ we mean a vertex subset 
$S \subseteq \{1,2,\ldots,n\}$ such that every hyperedge 
meets at most one vertex in $S$. 
\begin{theorem}
	Suppose that we are given a stable set $S$ of ${\cal H}(A)$.
	For $y \in \ZZ^n$, 
	we can compute $f_A(y)$ in $O(\poly(n,m,\|y\|_{\infty}) (\|y\|_{\infty}+1)^{n - |S|})$ time and   $O(\poly(n,m,\|y\|_{\infty})$ space.
\end{theorem}
\begin{proof}
	Let $N =\|y\|_{\infty} +1$, $d \coloneqq {\bf 1}^{\top} y$, and $\nu \coloneqq |S|$.
	As before, it suffices to compute \eqref{eqn:ts} modulo $(t^{d+1})$.
	We show that \eqref{eqn:ts} admits the following factorization:
	\begin{equation}
	\sum_{j} \sum_{h} 
	t^{\bm{1}^{\top} Ah} s^{j^{\top} (y- Ah)} 
	= \sum_{j' \in \{ 0,1,\ldots,N-1 \}^{n-\nu}} 
	F_0(j') F_1(j') \cdots F_{\nu}(j'), \label{eqn:factor} 
	\end{equation}
	where each $F_\alpha(j')$ is computable modulo $(t^{d+1})$ 
	in $O(\poly(n,m,N))$ time and space.
		
	 By arranging indices of $A$, 
	we can assume that $S = \{1,2,\ldots,\nu\}$, 
	and that $A$ is regarded as a block matrix $A  = (A^0\ A^1\ A^2\ \cdots A^{\nu})$, where 
    $A_\alpha$ $( \alpha = 1,2,\ldots,\nu)$
    consists of columns such that the corresponding hyperedge meets $\alpha \in S$ $(\alpha = 1,2,\ldots,\nu)$.
    Accordingly, vector $h \in \{0,1,2\ldots N-1\}^m$ is also partitioned as
    \[
    h = \left(\begin{array}{c}
    h^0 \\
    h^1 \\
    \vdots \\
    h^{\nu}
    \end{array}\right), \quad Ah = A^0h^0 + A^1 h^1 + \cdots + A^\nu h^\nu.
    \]
  
    We suppose that an $n - \nu$-dimensional vector $j' \in \{0,1,2\ldots,N-1\}^{n- \nu}$
   is embedded to $\{0,1,2\ldots,N-1\}^{n}$ 
   by filling $0$ to the first $\nu$ components.
   Each $j \in \{0,1,\ldots,N-1\}^n$ is uniquely represented as 
   $j = j' + \sum_{\alpha=1}^{\nu} j_\alpha e_{\alpha}$ for $j' \in \{0,1,2\ldots,N-1\}^{n- \nu}$, where $e_{\alpha}$ is the $\alpha$-th unit vector.
   Then we have
   \begin{eqnarray}
   j^{\top} y & = & {j'}^{\top} y + j_1 y_1 + \cdots + j_\nu y_{\nu}, \\
   j^{\top} A^{\alpha} h^{\alpha} &=&
   \left\{
   \begin{array}{ll}
   {j'}^{\top} A^{0} h^{0}  & {\rm if}\ \alpha = 0, \\
    ({j'} + j_{\alpha}e_{\alpha})^{\top} A^{\alpha} h^{\alpha} & {\rm if}\ \alpha > 0.
    \end{array}
   \right.
   \end{eqnarray}
   Define $G_0({j'}, h^0)$ and $G_\alpha({j'}, j_{\alpha}, h^{\alpha})$ $(\alpha= 1,2,\ldots, \nu)$ by
   \begin{eqnarray}
   G_0({j'}, h^0) &\coloneqq & t^{{\bf 1}^{\top} A^0h^0} s^{-{j'}^{\top} A^0h^0}, \\
       G_\alpha({j'}, j_{\alpha}, h^\alpha) &\coloneqq &  t^{{\bf 1}^{\top} A^\alpha h^\alpha} s^{- ({j'} + j_{\alpha} e_{\alpha})^{\top}A^\alpha h^\alpha}.
   \end{eqnarray}
   Then  
   \[
   t^{\bm{1}^{\top} Ah} s^{- j^{\top}Ah} = G_0({j'}, h^0) G_1 ({j'},j_1,h^1) \cdots  G_\nu ({j'},j_\nu,h^\nu). 
    \] 
    From $\sum_{j} \sum_{h} = \sum_{j'} \sum_{j_1} \sum_{j_2} \cdots \sum_{j_{\nu}} \sum_{h^0} \cdots \sum_{h^{\nu}}$, we see that the left hand side of \eqref{eqn:factor} is equal to
   \begin{equation}
   \sum_{j' \in \{0,1,\ldots,N-1\}^{n- \nu}} s^{{j'}^{\top}y} 
   \left(\sum_{h^0} G_0(j', h^0) \right)
   \prod_{\alpha=1}^{\nu} \sum_{j_{\alpha}=0}^{N-1} s^{j_\alpha y_{\alpha}}\sum_{h^{\alpha}} G_{\alpha} ({j'}, j_{\alpha}, h^{\alpha}), 
   \end{equation}
   where $h_{\alpha}$ ranges over $\{0,1,\ldots,N-1\}^{\nu_{\alpha}}$ 
   and $\nu_{\alpha}$ is the dimension of $h_{\alpha}$.
   Now $G_{\alpha}$ is a form of $t^{a^{\top}h^{\alpha}} s^{b^{\top}h^{\alpha}}$, and hence 
   $\sum_{h^{\alpha}} G_{\alpha}$ is factorized as 
   	$\prod_{l=1}^{\nu_{\alpha}} \sum_{k=0}^{N-1} t^{a_l k} s^{b_l k}$ (as in \eqref{poly_formula0}). 
   Thus each $\sum_{h_{\alpha}} G_{\alpha}$ is computable in 
   $O(\poly(n,m,N))$ time, and we have the desired expression \eqref{eqn:factor}.
\end{proof}
A hypergraph ${\cal H} = (V, {\cal E})$ is said to be {\em $k$-partite} if
there is a partition of the vertex set into $k$ nonempty subsets $S_1,S_2,\ldots,S_k$
such that each hyperedge meets at most one vertex in each $S_i$.
Clearly some $S_i$ has cardinality at least $|V|/k$. 
Hence we obtain a generalization of Theorem~\ref{thm:main2}.
\begin{corollary}
	Suppose that ${\cal H}(A)$ is $k$-partite and the partition is given.
	Then we can compute $f_A(y)$ in $O(\poly(n,m,\|y\|_\infty) (\|y\|_\infty + 1)^{(1 -1/k) n})$ time and $O(\allowbreak \poly(n,m,\|y\|_\infty)$ space.
\end{corollary}

Finally we note a combinatorial inclusion-exclusion formula 
for the number of perfect matchings derived from our formula (Corollary~\ref{corollary:integerpointsinpolytopes}).
A {\em perfect matching} in a hypergraph is a subset $M$
of hyperedges such that 
each vertex belongs to exactly one hyperedge in $M$.
\begin{corollary}\label{theorem:hypergraphmatching1}
    Let ${\cal H} = \left( V, \mathcal{E} \right)$ be a hypergraph. 
The number of perfect matchings in ${\cal H}$ is equal to the coefficient of $t^{|V|}$ in
\begin{equation}
    \frac{1}{2^{\lvert V \rvert}} \sum_{U \subseteq V} \left(-1\right)^{\lvert U \rvert} \prod_{F \in \mathcal{E}}
    \left( 1 +  t^{\lvert F \rvert}\left(-1\right) ^{ \lvert  F\cap U\rvert }\right).
    \label{eqn:formula_hyper}
\end{equation}
\end{corollary}
\begin{proof}
	Let $A \in \left\{ 0 , 1 \right\}^{\lvert V\rvert \times \lvert \mathcal{E} \rvert}$ be the adjacency matrix of the hypergraph 
    ${\cal H} = ( V, \mathcal{E})$. 
    Then a perfect matching is exactly a solution $x$ of $Ax = {\bf 1}$.
    Apply Corollary \ref{corollary:integerpointsinpolytopes} with $N = 2$.
    Then $\omega_2(k) = (-1)^{k}$, and
    $f_A({\bf 1})$ is equal to the coefficient of $t^{|V|}$ in
        \[
        \frac{1}{2^{\lvert V \rvert}} \sum_{j\in\left\{ 0, 1 \right\}^{\lvert V \rvert}}
        \left( -1 \right)^{\bm{1}^{\top} j} \prod_{l=1}^{\lvert \mathcal{E} \rvert}
        \left( 1 + t^{-\bm{1}^{\top} A_l } \left(-1\right)^{- j^{\top} A_{l}}\right).
        \]
    Identify $j \in \{0,1\}^{|V|}$ with a subset $U \subseteq V$.
    Then $j^{\top}{\bf 1} = |U|$, ${\bf 1}^{\top} A_{l} = |F|$ and $j^{\top} A_l = |U \cap F|$, where $F$ is the hyperedge corresponding to $l$-th column $A_l$ of $A$.
    Thus we have the formula.
\end{proof}
A hypergraph is said to be {\em $\ell$-uniform} if 
each hyperedge has cardinality $\ell$.
A $2$-uniform hypergraph is exactly a simple undirected graph. 
In the case of a uniform hypergraph, the coefficient of $t^{|V|}$ in (\ref{eqn:formula_hyper})
is the following simple expression. 
\begin{corollary}
    \label{cor:kuniform}
    Let ${\cal H} = \left( V, \mathcal{E} \right)$  be an $\ell$-uniform hypergraph.
    The number of perfect matchings in ${\cal H}$ is equal to
    \begin{equation}
        \frac{1}{2^n} \sum_{U \subseteq V} \left(-1\right)^{\lvert U \rvert}
        \sum_{i=0}^{|V|/\ell}\left(-1\right)^i 
        \binom{|{\cal E}_{U,{\rm odd}}|}{i}
        \binom{|{\cal E} \setminus {\cal E}_{U,{\rm odd}}|}{ |V|/\ell - i}, \label{eqn:formula_uniform}
    \end{equation}
    where ${\cal E}_{U, {\rm odd}}$ 
    denotes the subsets of ${\cal E}$ consisting of $F$ 
    with $|F \cap U|$ odd.
\end{corollary}

\begin{proof}
    The formula in Corollary~\ref{theorem:hypergraphmatching1} becomes 
    \begin{equation}
    \frac{1}{2^{\lvert V \rvert}} \sum_{U \subseteq V} \left(-1\right)^{\lvert U \rvert} \prod_{F \in \mathcal{E}_{U,{\rm odd}}}
    ( 1 -  t^{\ell})  
    \prod_{G \in {\cal E} \setminus \mathcal{E}_{U,{\rm odd}}}
    ( 1 +  t^{\ell}).
    \end{equation}
    By evaluating the coefficient of $t^{|V|}$, we obtain the formula.    
\end{proof}
A similar inclusion-exclusion 
formula for hypergraph matching is given in \cite{bjorklund2008exact}.
In the case of a graph, i.e., $\ell = 2$, the formula (\ref{eqn:formula_uniform}) becomes
    \begin{equation}
    \frac{1}{2^n} \sum_{U \subseteq V} \left(-1\right)^{\lvert U \rvert}
    \sum_{i=0}^{|V|/2}\left(-1\right)^i 
    \binom{|\delta U|}{i}
    \binom{|{\cal E} \setminus \delta U|}{ |V|/2 - i},
    \end{equation}
where $\delta U$ denotes the set of edges for which exactly one of ends belongs to $U$. 

\section{Concluding Remarks}

In this paper, we presented a new algorithm for 
counting integer points in polytopes.
Our original attempt was to count integer points 
by computing the inverse Z-transformation directly by 
numerical integration in floating-point arithmetic.
Although this approach did not work well, 
the theoretical analysis on the error estimate  
brought a new inclusion-exclusion formula for integer points, 
on which our algorithm is built. 

We end this paper with some open problems and future work:
\begin{itemize}
    \item[$\bullet$] We employed the trapezoidal rule for the numerical integration of the inverse Z-trans\-for\-ma\-tion.
    Can other (more sophisticated) methods of numerical integration and their error analysis lead to a better algorithm for integer point counting ?
	
	\item[$\bullet$]  For exact counting of 
	perfect matchings in general $n$-vertex graphs,  
	the current fastest (polynomial space) algorithms are $O^*(2^{n/2})$-time algorithms by
	Bj\"orklund~\cite{bjorklund2012counting} and Cygan and Pilipczuk~\cite{cygan2015faster}. This time complexity matches one by Ryser for bipartite graphs.
	
	For counting $b$-matchings in general graphs, 
	our algorithm in Theorem~\ref{thm:main} brings
	an $O^*(\poly(b) (\|b\|_{\infty}+1)^{n})$-time algorithm, which is 
	improved to $O^*(\poly(b) (\|b\|_{\infty}+1)^{n/2})$-time one 
	for bipartite graphs (Theorem~\ref{thm:main2}). 
	So a natural question is: Can we design a polynomial space $O^*(\poly(b) (\|b\|_{\infty}+1)^{n/2})$-time algorithm for counting $b$-matchings in general graphs ?
	  
	\item[$\bullet$] Our algorithm is simple, and is not difficult to be implemented.
	Implementing our algorithm, evaluating its performance compared with \allowbreak LattE~\cite{DELOERA20041273} (and other lattice counting problems), 
	and incorporating heuristics for speeding up deserve interesting future research. 
\end{itemize}

\section*{Acknowledgements}
We thank the anonymous referee for helpful comments.
H. Hirai is supported by the grant-in-aid of Japan Society of the Promotion of Science with KAKENHI Grant Numbers JP26280004, JP17K00029.
K. Tanaka is supported by the grant-in-aid of Japan Society of the Promotion of Science with KAKENHI Grant Number 17K14241.

\bibliography{counting_integral_0714}

\end{document}